\documentclass[letterpaper,11pt]{article}

\usepackage{times}

\usepackage{amsthm,amsmath,amsfonts}

\usepackage{graphicx}
\usepackage[small]{caption}
\usepackage[hidelinks]{hyperref}
\newcommand{\orcid}[1]{\,\href{https://orcid.org/#1}{\includegraphics[width=8pt]{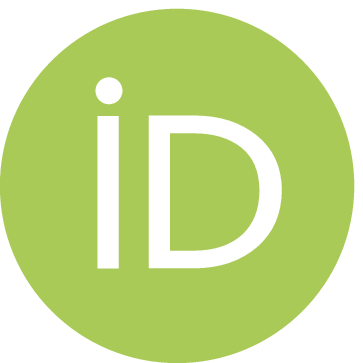}}}

\usepackage{xcolor}
\newcommand{\hl}[1]{\textcolor{violet}{#1}}
\newcommand{\red}[1]{\textcolor{red}{#1}}
\newcommand{\blue}[1]{\textcolor{blue}{#1}}
\newcommand{\true}[1]{\textcolor{red}{#1}}
\newcommand{\false}[1]{\textcolor{blue}{#1}}

\newcommand{\I}{\mathcal{I}}
\newcommand{\C}{\mathcal{C}}

\newcommand{\mis}[1]{\textsc{Max-IS-$#1$}}
\newcommand{\sat}[2]{\textsc{E$#1$-Occ-Max-E$#2$-SAT}}
\newcommand{\maxadj}{\textsc{Max-Adj}}
\newcommand{\minbrk}{\textsc{Min-Brk}}
\newcommand{\opt}{\mathrm{opt}}
\newcommand{\val}{\mathrm{val}}
\newcommand{\po}{\mbox{\footnotesize\textsl{partial}}}
\newcommand{\io}{\mbox{\footnotesize\textsl{interval}}}
\newcommand{\so}{\mbox{\footnotesize\textsl{semi}}}
\newcommand{\wo}{\mbox{\footnotesize\textsl{weak}}}
\newcommand{\lo}{\mbox{\footnotesize\textsl{linear}}}
\newcommand{\pp}{$(\po,\po)$}
\newcommand{\ww}{$(\wo,\wo)$}
\newcommand{\lp}{$(\lo,\po)$}
\newcommand{\li}{$(\lo,\io)$}
\newcommand{\lw}{$(\lo,\wo)$}
\newcommand{\ls}{$(\lo,\so)$}

\newtheorem{theorem}{Theorem}
\newtheorem{lemma}{Lemma}
\newtheorem{proposition}{Proposition}

\title{Partial order alignment by adjacencies and breakpoints}

\author{Rain Jiang\orcid{0000-0002-0144-942X}\qquad
Kai Jiang\orcid{0000-0001-8165-0571}\qquad
Minghui Jiang\orcid{0000-0003-1843-9292}\,\thanks{\texttt{ dr.minghui.jiang at gmail.com}}\medskip\\
Home School, USA}

\date{}

\begin{document}

\maketitle

\begin{abstract}
Linearizing two partial orders
to maximize the number of adjacencies and minimize the number of breakpoints
is APX-hard.
This holds even if one of the two partial orders
is already a linear order and the other is an interval order,
or if both partial orders are weak orders.
\end{abstract}

\section{Introduction}

A \emph{partial order} $\prec$
on a set $\Sigma$ of elements
is a binary relation on certain pairs of elements in $\Sigma$,
such that for all $a,b,c\in \Sigma$,
\begin{itemize}\setlength\itemsep{0pt}
\item
\emph{irreflexivity}: not $a \prec a$, that is, no element is related to itself,
\item
\emph{transitivity}: if $a \prec b$ and $b \prec c$, then $a \prec c$,
\item
\emph{asymmetry}: if $a \prec b$ then not $b \prec a$.
\end{itemize}

With respect to a partial order $\prec$ on $\Sigma$,
we say that two distinct elements $a$ and $b$ in $\Sigma$ are \emph{incomparable}
if neither $a \prec b$ nor $b \prec a$.
A partial order $\prec$ on $\Sigma$ is a \emph{linear order}
if for all $a,b \in \Sigma$, either $a = b$,
or $a \prec b$, or $b \prec a$.
A \emph{linear extension} of a partial order $\prec$ on $\Sigma$
is a linear order $\prec'$
on the same set $\Sigma$ of elements,
such that for all $a,b\in \Sigma$, $a \prec b \implies a \prec' b$.

A partial order $\prec$ on $\Sigma$ can be represented as
a directed acyclic graph $G$ with vertex set $\Sigma$
such that $u \prec v$ if and only if there is a directed path from $u$ to $v$ in $G$.
Then a linear extension of $\prec$ corresponds to a topological sort of $G$.

A linear order $\prec$ on $\Sigma$ is commonly represented by
a unique permutation $\pi$ of $\Sigma$
such that $a \prec b$ if and only if $a$ precedes $b$ in $\pi$.
Henceforth when we denote a linear order by $\prec$,
we also use the same symbol $\prec$
to refer to the unique permutation representing the linear order.

\medskip
In this paper, we study the problem of partial order alignment.
Specifically,
given two partial orders $\Gamma$ and $\Pi$ on the same set $\Sigma$ of elements,
we want to linearize the two partial orders $\Gamma$ and $\Pi$
into two linear orders $\Gamma'$ and $\Pi'$, respectively,
such that the two permutations $\Gamma'$ and $\Pi'$ are as similar as possible,
by certain genome rearrangement measures.

In comparative genomics,
a genomic map can be represented by a partial order on
a set $\Sigma$ of markers that annotate genomes.
For two markers $a$ and $b$ in $\Sigma$,
and for two permutations $\pi'$ and $\pi''$ of $\Sigma$,
the ordered pair $(a,b)$ is called an \emph{adjacency} of $\pi'$ and $\pi''$
if $a$ appears immediately before $b$ in both $\pi'$ and $\pi''$.
An ordered pair of consecutive elements in $\pi'$,
if not an adjacency of $\pi'$ and $\pi''$,
is called a \emph{breakpoint} of $\pi'$ with respect to $\pi''$.
Note that the number of breakpoints of $\pi'$ with respect to $\pi''$
is the same as the number of breakpoints of $\pi''$ with respect to $\pi'$,
which is equal to the number of markers in $\Sigma$ minus one then
minus the number of adjacencies of $\pi'$ and $\pi''$.

Let \maxadj\ (respectively, \minbrk)
be the problem of
linearizing two given partial orders $\Gamma$ and $\Pi$
on the same set $\Sigma$ of markers
into two linear orders $\Gamma'$ and $\Pi'$, respectively,
such that the number
$n_\textsl{adj}$
of adjacencies of the two permutations $\Gamma'$ and $\Pi'$
(respectively, the number
$n_\textsl{brk}$
of breakpoints of $\Gamma'$ with respect to $\Pi'$)
is maximized (respectively, minimized).
With a solution to either problem,
the two numbers
$n_\textsl{adj}$
and
$n_\textsl{brk}$
can then be used to measure
the similarity and the distance, respectively,
of the two genomic maps represented by $\Gamma$ and $\Pi$.

\medskip
An \emph{interval graph} is the intersection graph $G$ of a family $\I$ of open intervals,
with one vertex for each interval in $\I$,
and with an edge between two vertices if and only if the corresponding two intervals intersect.
Here $\I$ is called a \emph{representation} of $G$.

For two intervals $I$ and $J$,
we say that $I$ \emph{precedes} $J$,
if $I$ is disjoint from and to the left of $J$.
An \emph{interval order} is a partial order $\prec$ on a family $\I$ of open intervals,
such that for all $I,J\in\I$,
$I \prec J$ if and only if $I$ precedes $J$~\cite{Fi85}.
Here we also call $\I$ a \emph{representation} of $\prec$.

An interval graph is a \emph{proper interval graph}
if it has a representation
in which no interval properly contains another interval.
An interval graph is a \emph{unit interval graph}
if it has a representation
in which all intervals have the same length.
It is well known that an interval graph is a proper interval graph
if and only if it is a unit interval graph~\cite{BW99}.
A semiorder is a partial order on elements with numerical scores,
where elements with widely differing scores are ordered by their scores,
and where elements with close scores within a given margin of error
are deemed incomparable~\cite{Lu56}.
A \emph{semiorder} can be equivalently defined as
an interval order with a representation
in which all intervals have the same length~\cite{Fi85}.

A \emph{cluster graph} is a disjoint union of cliques.
A weak order can be seen as a relaxation
of a linear order where some elements may be tied with each other.
More precisely, a \emph{weak order} is a partial order $\prec$ on a set $\Sigma$
with a partition into $k$ subsets
$\Sigma_1,\Sigma_2,\ldots,\Sigma_k$
which we call \emph{buckets},
such that for $a \in \Sigma_i$ and $b \in \Sigma_j$,
$a \prec b$ if and only if $i < j$.
Note that elements in the same bucket are incomparable.
Just as an interval order is analogous to an interval graph,
and as a semiorder is analogous to a unit\,/\,proper interval graph,
a weak order is analogous to a cluster graph.
A partial order (respectively, a graph) is
a weak order (respectively, a cluster graph)
if and only if
it is an interval order (respectively, an interval graph)
with a representation in which
all intervals have length $1$
and have integer endpoints.

Denote by
$\po$ the class of all partial orders,
and denote by $\lo$ the class of all linear orders.
Similarly, denote by $\io$, $\so$, and $\wo$
the classes of interval orders, semiorders,
and weak orders, respectively.
Then we have the following hierarchy:
$$
\lo
\;\subseteq\;
\wo
\;\subseteq\;
\so
\;\subseteq\;
\io
\;\subseteq\;
\po
$$

For $\C_1,\C_2 \in \{ \lo, \wo, \so, \io, \po \}$,
denote by
\maxadj$(\C_1,\C_2)$ and \minbrk$(\C_1,\C_2)$,
respectively,
the two problems \maxadj\ and \minbrk\
on a partial order of class $\C_1$ and a partial order of class $\C_2$.
In general,
the problem \minbrk\pp\ is NP-hard~\cite{FJ07}.
Moreover, \minbrk\lp\ is NP-hard~\cite{BBHGBE07}
and even APX-hard~\cite{BFR13}.

We obtain the following results:

\begin{theorem}\label{thm:li}
\maxadj\li\ and \minbrk\li\
are APX-hard.
\end{theorem}

\begin{theorem}\label{thm:ww}
\maxadj\ww\ and \minbrk\ww\
are APX-hard.
This holds even if every bucket of the two weak orders
has at most two elements.
\end{theorem}

\begin{proposition}\label{prp:lw}
\maxadj\lw\ and \minbrk\lw\ admit a polynomial-time exact algorithm.
\end{proposition}

Our definition of partial orders
is one of two common definitions which are slightly different.
All problems and results in this paper can be equivalently formulated in terms of
\emph{non-strict} partial orders
instead of
\emph{strict} partial orders defined here.

\section{Preliminaries}

Given two optimization problems X and Y,
an \emph{L-reduction}~\cite{PY91} from X to Y
consists of two polynomial-time functions $f$ and $g$
and two positive constants $\alpha$ and $\beta$
satisfying the following two properties:
\begin{enumerate}\setlength\itemsep{0pt}
\item
For every instance $x$ of X, $f(x)$ is an instance of Y such that
\begin{equation}\label{eq:f}
\opt(f(x)) \le \alpha\cdot \opt(x).
\end{equation}
\item
For every feasible solution $y$ to $f(x)$, $g(y)$ is a feasible solution to $x$ such that
\begin{equation}\label{eq:g}
|\opt(x) - \val(g(y))|
\le \beta\cdot |\opt(f(x)) - \val(y)|.
\end{equation}
\end{enumerate}
Here $\opt(x)$ denotes the value of the optimal solution to an instance $x$,
and $\val(y)$ denotes the value of a solution $y$.
The two properties of an L-reduction imply the following inequality
on the relative errors of approximation:
$$
\frac{|\opt(x) - \val(g(y))|}{\opt(x)}
\le \alpha\beta\cdot
\frac{|\opt(f(x)) - \val(y)|}{\opt(f(x))}.
$$
Thus if there is an L-reduction from X to Y,
and if X is NP-hard to approximate within some constant relative error $\epsilon$,
then Y is NP-hard to approximate within a constant relative error $\epsilon/(\alpha\beta)$.

\section{APX-hardness of aligning a linear order and an interval order}

In this section we prove Theorem~\ref{thm:li}.
We prove the APX-hardness of the two problems
\maxadj\li\ and \minbrk\li\
by two L-reductions (based on the same construction)
from the APX-hard problem
\mis{3}, i.e.,
\textsc{Maximum Independent Set} in graphs of maximum degree $3$~\cite{AK00,CC06}.

Let $G$ be a graph of maximum degree $3$,
with $n$ vertices $\{1,2,\ldots,n\}$ and $m$ edges $\{1,2,\ldots,m\}$.
We will construct a set $\Sigma$ of markers,
a linear order $\Gamma$, and an interval order $\Pi$ in the following.

\medskip
There are $3n + 4m$ markers in the set $\Sigma$:
\begin{itemize}\setlength\itemsep{0pt}
\item
$2n$ markers including two vertex markers $u_i$ and $v_i$
for each vertex $i$, $1 \le i \le n$.
\item
$3m$ markers including
two selection markers $p_j$ and $q_j$ and an edge marker $e_j$
for each edge $j$, $1 \le j \le m$.
\item
$n + m$ separation markers $z_h$, $1 \le h \le n + m$.
\end{itemize}

The linear order $\Gamma$ is represented by the following permutation of $\Sigma$:
$$
\Gamma:
\quad
u_1 v_1
\;\;z_1\;\;
u_2 v_2
\;\;z_2\;\;
\ldots\;\;
u_n v_n
\;\;z_n\;\;
\quad
p_1 e_1 q_1
\;\;z_{n+1}\;\;
p_2 e_2 q_2
\;\;z_{n+2}\;\;
\ldots\;\;
p_m e_m q_m
\;\;z_{n+m}
.
$$

For each edge $j$, $1 \le j \le m$,
denote by $l_j$ and $r_j$, $1 \le l_j < r_j \le n$,
the two vertices incident to $j$.
Conversely, for each vertex $i$, $1 \le i \le n$,
\begin{itemize}\setlength\itemsep{0pt}
\item
denote by $\sigma_i$ the number of edges $j$ with $l_j = i$,
and denote by $\subset_{i,1},\subset_{i,2},\ldots,\subset_{i,\sigma_i}$ these $\sigma_i$ edges,
\item
denote by $\tau_i$ the number of edges $j$ with $r_j = i$,
and denote by $\supset_{i,1},\supset_{i,2},\ldots,\supset_{i,\tau_i}$ these $\tau_i$ edges.
\end{itemize}

To construct the interval order $\Pi$,
we first construct a sequence $Z$ of the markers in $\Sigma$:
$$
Z:
\quad
z_1 z_2 \ldots z_n z_{n+1} \ldots z_{n+m}
\quad
\langle u_1 v_1 \rangle\;
\langle u_2 v_2 \rangle\;
\ldots\;
\langle u_n v_n \rangle,
$$
where $\langle u_i v_i \rangle$ for $1 \le i \le n$ is a subsequence of markers
$$
\langle u_i v_i \rangle:
\quad
\blue{q_{\subset_{i,1}}
\ldots
q_{\subset_{i,\sigma_i}}}
\quad
u_i
\quad
\blue{p_{\supset_{i,1}} e_{\supset_{i,1}}
\ldots
p_{\supset_{i,\tau_i}} e_{\supset_{i,\tau_i}}}
\quad
\red{e_{\subset_{i,1}} q_{\subset_{i,1}}
\ldots
e_{\subset_{i,\sigma_i}} q_{\subset_{i,\sigma_i}}}
\quad
v_i
\quad
\red{p_{\supset_{i,1}}
\ldots
p_{\supset_{i,\tau_i}}}
\;
.
$$

Note that for each edge $j$, $1 \le j \le m$,
each of the three markers $p_j$, $q_j$ and $e_j$ occurs twice in $Z$:
\begin{itemize}\setlength\itemsep{0pt}
\item
$p_j$ occurs
twice in $\langle u_i v_i \rangle$ for $i = r_j$,
\item
$q_j$ occurs
twice in $\langle u_i v_i \rangle$ for $i = l_j$,
\item
$e_j$ occurs
once in $\langle u_i v_i \rangle$ for $i = l_j$,
and once in $\langle u_i v_i \rangle$ for $i = r_j$.
\end{itemize}

\begin{figure}[htbp]
\centering\includegraphics{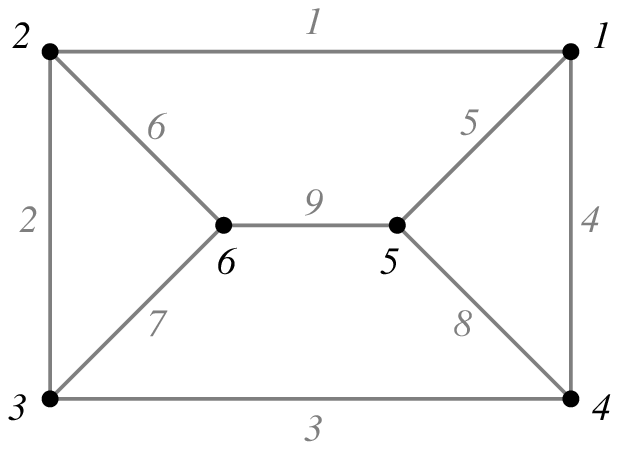}
$$
\begin{array}{ccccccc}
\langle u_1 v_1 \rangle: &
\blue{q_1 \; q_4 \; q_5} & u_1 & & \red{e_1 q_1 \; e_4 q_4 \; e_5 q_5} & v_1 & \\
\langle u_2 v_2 \rangle: &
\blue{q_2 \; q_6} & u_2 & \blue{p_1 e_1} & \red{e_2 q_2 \; e_6 q_6} & v_2 & \red{p_1}\\
\langle u_3 v_3 \rangle: &
\blue{q_3 \; q_7} & u_3 & \blue{p_2 e_2} & \red{e_3 q_3 \; e_7 q_7} & v_3 & \red{p_2}\\
\langle u_4 v_4 \rangle: &
\blue{q_8} & u_4 & \blue{p_3 e_3 \; p_4 e_4} & \red{e_8 q_8} & v_4 & \red{p_3 \; p_4}\\
\langle u_5 v_5 \rangle: &
\blue{q_9} & u_5 & \blue{p_5 e_5 \; p_8 e_8} & \red{e_9 q_9} & v_5 & \red{p_5 \; p_8}\\
\langle u_6 v_6 \rangle: &
& u_6 & \blue{p_6 e_6 \; p_7 e_7 \; p_9 e_9} & & v_6 & \red{p_6 \; p_7 \; p_9}
\end{array}
$$
\caption{A cubic graph $G$ and the corresponding subsequences $\langle u_i v_i \rangle$ of $Z$.}
\label{fig:mis}
\end{figure}

Refer to Figure~\ref{fig:mis} for an example.

For each marker that occurs only once in $Z$,
replace it by two consecutive copies.
Let $Z_2$ be the resulting sequence in which every marker in $\Sigma$ occurs exactly twice.
For each marker in $\Sigma$,
construct an interval between the two integer indices of its two occurrences in $Z_2$.
Let $\I$ be the family of intervals thus constructed,
including one interval for each marker in $\Sigma$.
Then $\Pi$ is represented by $\I$.
This completes the construction.

\begin{lemma}\label{lem:li}
$G$ has an independent set of at least $k$ vertices
if and only if
$\Pi$ admits a linear extension
having at least $m + k$ adjacencies with $\Gamma$.
\end{lemma}

\begin{proof}
We first prove the direct implication.
Suppose $G$ has an independent set $V'$ of $k$ vertices.
For each edge $j$ incident to two vertices $a$ and $b$,
where $1 \le j \le m$ and $1 \le a < b \le n$,
the relative positions of the double occurrences
of the three markers $p_j,q_j,e_j$
among the four markers $u_a,v_a,u_b,v_b$
in $Z$ are as follows:
$$
Z:
\quad
\;\ldots\;
\blue{q_j}
\;\ldots\;
u_a
\;\ldots\;
\red{e_j q_j}
\;\ldots\;
v_a
\;\ldots\;
u_b
\;\ldots\;
\blue{p_j e_j}
\;\ldots\;
v_b
\;\ldots\;
\red{p_j}
\;\ldots\;
.
$$
If $a \in V'$, remove
$\red{e_j q_j}$
and
$\red{p_j}$;
otherwise, remove
$\blue{q_j}$
and
$\blue{p_j e_j}$.
Then the resulting subsequence $Z'$ of $Z$ is a permutation of $\Sigma$
and a linear extension of $\Pi$.
Moreover, between $\Gamma$ and $Z'$,
we have an adjacency $u_i v_i$ for each vertex $i \in V'$,
and an adjacency either $\blue{p_j e_j}$ or $\red{e_j q_j}$ for each edge $j$, $1 \le j \le m$.
The total number of adjacencies is at least $m + k$.

We next prove the reverse implication.
Suppose that $\Pi$ admits a linear extension $\Pi'$
having at least $m + k$ adjacencies with $\Gamma$.
Because of the separation markers,
the only possible adjacencies of $\Gamma$ and $\Pi'$
are $p_j e_j$ and $e_j q_j$ for $1 \le j \le m$,
and $u_i v_i$ for $1 \le i \le n$.
Moreover,
we cannot have both adjacencies $p_j e_j$ and $e_j q_j$ for the same $j$,
because the interval for $q_j$ precedes the interval for $p_j$ in $\Pi$,
opposite to their order in $\Gamma$.
Thus there can be at most $m$ adjacencies
among $p_j e_j$ and $e_j q_j$ for $1 \le j \le m$.
Thus the at least $m + k$ adjacencies of $\Gamma$ and $\Pi'$
include at least $k$ adjacencies among $u_i v_i$ for $1 \le i \le n$.

Suppose that these adjacencies include both $u_a v_a$ and $u_b v_b$
for two vertices $a$ and $b$, $1 \le a < b \le n$,
and there is an edge $j$, $1 \le j \le m$, incident to both $a$ and $b$.
Note that in $\Pi$, the interval for $u_a$ precedes the interval for $e_j$,
and the interval for $q_j$ precedes the interval for $v_a$.
Since $u_a v_a$ is an adjacency of $\Gamma$ and $\Pi'$,
we must have $q_j$ before $u_a v_a$, and $e_j$ after $u_a v_a$ in $\Pi'$.
Thus $e_j q_j$ cannot be an adjacency.
Similarly,
since $u_b v_b$ is an adjacency,
$p_j e_j$ cannot be an adjacency.
Now remove $e_j$ and $q_j$ from $\Pi'$
and insert them back between $u_a$ and $v_a$,
thereby destroying the adjacency $u_a v_a$ but creating a new adjacency $e_j q_j$.
Then we obtain another linear extension of $\Pi$ having at least as many adjacencies
with $\Gamma$.
By this replacement argument, we can assume without loss of generality
that if the adjacencies of $\Gamma$ and $\Pi'$
include both $u_a v_a$ and $u_b v_b$
for two different vertices $a$ and $b$,
then there can be no edge incident to both $a$ and $b$.
Thus the subset of at least $k$ vertices $i$
for all adjacencies $u_i v_i$ of $\Gamma$ and $\Pi'$
is an independent set in $G$.
\end{proof}

We are now ready to present two L-reductions from \mis{3} to
\maxadj\li\ and \minbrk\li, respectively.
Recall~\eqref{eq:f} and~\eqref{eq:g}.

Since $G$ has maximum degree~$3$,
each vertex is incident to at most three edges.
On the other hand, each edge is incident to exactly two vertices.
By double counting the number of vertex-edge incidences,
we have $2m \le 3n$,
and hence $m \le \frac32 n$.

Let $k^*$ be the maximum number of vertices in an independent set in $G$.
By a greedy algorithm that repeatedly removes a vertex and its (at most $3$) adjacent vertices,
we can obtain an independent set of at least $n / 4$ vertices.
Thus $k^* \ge n / 4$, and hence $n \le 4 k^*$.

In both L-reductions,
the polynomial-time function $f$
is just the construction described above which satisfies Lemma~\ref{lem:li}.
By the direct implication of Lemma~\ref{lem:li},
the maximum number of adjacencies is
$$
m + k^* \le \frac32n + k^* \le 6 k^* + k^* = 7 k^*.
$$
Thus we can set $\alpha = 7$ for the L-reduction to \maxadj\li.

Since the number of markers is $3n + 4m$,
the minimum number of breakpoints is
$$
(3n + 4m - 1) - (m + k^*)
\le 3n + 3m - k^*
\le 3n + 3\cdot\frac32 n - k^*
= \frac{15}2 n - k^*
\le 30 k^* - k^*
= 29 k^*.
$$
Thus we can set $\alpha = 29$ for the L-reduction to \minbrk\li.

For the other direction,
consider any solution to the reduced instance,
with $m + k$ adjacencies and $(3n + 4m - 1) - (m + k)$ breakpoints.
By the reverse implication of Lemma~\ref{lem:li},
there is a function $g$ that converts this solution
to an independent set of at least $k$ vertices in $G$.
The absolute error of the solution, for both \maxadj\li\ and \minbrk\li,
is exactly $|k^* - k|$.
The absolute error of the converted solution for \mis{3}
is at most $|k^* - k|$.
Thus we can set $\beta = 1$ for both L-reductions.

Since \mis{3} is APX-hard, it follows by the two L-reductions that
the two optimization problems \maxadj\li\ and \minbrk\li\ are APX-hard too.
This completes the proof of Theorem~\ref{thm:li}.

\section{APX-hardness of aligning two weak orders}

In this section we prove Theorem~\ref{thm:ww}.
We prove the APX-hardness of the two problems
\maxadj\ww\ and \minbrk\ww\
by two L-reductions (based on the same construction)
from the APX-hard problem \sat{3}{2}~\cite{BK99,BK03}.

Given a set $X$ of $n$ variables and a set $C$ of $m$ clauses,
where each variable has exactly $3$ literals
(in $3$ different clauses)
and each clause is the disjunction of exactly $2$ literals
(of $2$ different variables),
\sat{3}{2} is the problem of finding an assignment of $X$ that satisfies
the maximum number of clauses in $C$.
The problem \sat{3}{2} is known to be APX-hard even if 
the $3$ literals of each variable are neither all positive nor all negative;
see for example the gap-preserving reduction from
\textsc{Max-Cut} in cubic graphs to this problem
in the DIMACS version of~\cite{BK99}.
As a result, we can assume that each variable has
either $2$ positive and $1$ negative literals,
or $1$ positive and $2$ negative literals.

Let $(X,C)$ be an instance of \sat{3}{2},
where $X$ is a set of $n$ variables $x_i$, $1 \le i \le n$,
and $C$ is a set of $m$ clauses $c_j = c_j^1 \lor c_j^2$, $1 \le j \le m$.
We will construct a set $\Sigma$ of markers,
and two weak orders $\Gamma$ and $\Pi$ on $\Sigma$,
in the following.

\medskip
The set $\Sigma$ includes $12n + 5m$ markers:
\begin{itemize}\setlength\itemsep{0pt}

\item
For each variable $x_i$, $1 \le i \le n$,
$\Sigma$ includes seven variable markers
$p_i, q_i, r_i, s_i, t_i, u_i, v_i$,
a pair of positive selection markers
$(a_i^+, b_i^+)$,
a pair of negative selection markers
$(a_i^-, b_i^-)$,
and a dummy marker $d_i$.

\item
For each clause $c_j$ consisting of two literals $c_j^1$ and $c_j^2$, $1 \le j \le m$,
$\Sigma$ includes
a pair of literal markers
$(e_j^1, f_j^1)$ for $c_j^1$,
a pair of literal markers
$(e_j^2, f_j^2)$ for $c_j^2$,
and a separation marker $z_j$.

\end{itemize}

The two weak orders $\Gamma$ and $\Pi$ are schematically represented as follows:
\begin{align*}
\Gamma&:
\quad
\langle C_1 \rangle
\quad
\{ z_1 \}
\quad
\langle C_2 \rangle
\quad
\{ z_2 \}
\quad
\ldots
\quad
\langle C_m \rangle
\quad
\{ z_m \}
\qquad
\langle X_1 \rangle
\quad
\langle X_2 \rangle
\quad
\ldots
\quad
\langle X_n \rangle
\\
\Pi&:
\quad
\langle Y_1 \rangle
\quad
\langle Y_2 \rangle
\quad
\ldots
\quad
\langle Y_n \rangle
\qquad
\{ z_1 \}
\quad
\{ z_2 \}
\quad
\ldots
\quad
\{ z_m \}
\;
.
\end{align*}

For each clause $c_j$, $1 \le j \le m$,
the clause gadget $\langle C_j \rangle$ in $\Gamma$
consists of two buckets:
$$
\langle C_j \rangle:
\quad
E_j
\quad
F_j
\;
.
$$

For each variable $x_i$, $1 \le i \le n$,
the variable gadget
$\langle X_i \rangle$ in $\Gamma$
consists of eight buckets:
$$
\langle X_i \rangle:
\quad
P'_i
\quad
Q'_i
\quad
R'_i
\quad
S'_i
\quad
T'_i
\quad
U'_i
\quad
V'_i
\quad
D_i
\;
.
$$

For each variable $x_i$, $1 \le i \le n$,
the selection gadget
$\langle Y_i \rangle$ in $\Pi$
consists of nine buckets:
$$
\langle Y_i \rangle:
\quad
P_i
\quad
Q_i
\quad
R_i
\quad
S_i
\quad
T_i
\quad
U_i
\quad
V_i
\quad
A_i
\quad
B_i
\;
.
$$

We first place the variable markers.
For each $i$, $1 \le i \le n$,
put
$p_i$ in both $P'_i$ and $P_i$,
$q_i$ in both $Q'_i$ and $Q_i$,
$r_i$ in both $R'_i$ and $R_i$,
$s_i$ in both $S'_i$ and $S_i$,
$t_i$ in both $T'_i$ and $T_i$,
$u_i$ in both $U'_i$ and $U_i$,
$v_i$ in both $V'_i$ and $V_i$.

We next place the selection markers.
For each $i$, $1 \le i \le n$,
\begin{itemize}\setlength\itemsep{0pt}
\item
put
$a_i^+$ in $P'_i$,
$b_i^+$ in $Q'_i$,
$a_i^-$ in $U'_i$,
$b_i^-$ in $V'_i$,
\item
put
$a_i^+$
and
$a_i^-$ in $A_i$,
$b_i^+$
and
$b_i^-$ in $B_i$.
\end{itemize}

We next place the dummy markers and literal markers in $\Gamma$:
\begin{itemize}\setlength\itemsep{0pt}
\item
For each $i$, $1 \le i \le n$,
put the dummy marker $d_i$ in $D_i$.
\item
For each $j$, $1 \le j \le m$,
put the two literal markers
$e_j^1$ and $e_j^2$
in $E_j$,
and put the two literal markers
$f_j^1$ and $f_j^2$
in $F_j$.
\end{itemize}

We then place the dummy markers and literal markers in $\Pi$.
Each variable $x_i$, $1 \le i \le n$,
has three literals $x_i^1, x_i^2, x_i^3$ in three clauses.
Without loss of generality,
assume that $x_i^1$ is positive and $x_i^3$ is negative.
Each literal $x_i^g$, $1 \le i \le n$ and $1 \le g \le 3$,
is $c_j^h$ for some $j$ and $h$, $1 \le j \le m$ and $1 \le h \le 2$,
which has two corresponding literal markers $e_j^h$ and $f_j^h$.
Put these six literal markers for $x_i$ and the dummy marker $d_i$
in the seven buckets
$P_i,Q_i,R_i,S_i,T_i,U_i,V_i$
as follows,
one marker in each bucket:
\begin{itemize}\setlength\itemsep{0pt}
\item
Put the two literal markers $e_j^h$ and $f_j^h$
of $x_i^1$ in $P_i$ and $Q_i$, respectively.
\item
Put the two literal markers $e_j^h$ and $f_j^h$
of $x_i^3$ in $U_i$ and $V_i$, respectively.
\item
For $x_i^2$,
\begin{itemize}\setlength\itemsep{0pt}
\item
if it is positive,
put the two literal markers $e_j^h$ and $f_j^h$
of $x_i^2$ in $R_i$ and $S_i$, respectively,
and put the dummy marker $d_i$ in $T_i$,
\item
if it is negative,
put the two literal markers $e_j^h$ and $f_j^h$
of $x_i^2$ in $S_i$ and $T_i$, respectively,
and put the dummy marker $d_i$ in $R_i$.
\end{itemize}
\end{itemize}

\begin{figure}[htbp]
\begin{align*}
\Gamma:
&\qquad
\{ e_1^1, e_1^2 \}
\;
\{ f_1^1, f_1^2 \}
\quad
\{ z_1 \}
\quad
\{ e_2^1, e_2^2 \}
\;
\{ f_2^1, f_2^2 \}
\quad
\{ z_2 \}
\quad
\{ e_3^1, e_3^2 \}
\;
\{ f_3^1, f_3^2 \}
\quad
\{ z_3 \}
\\
&\qquad
\{ p_1, a_1^+ \}\;
\{ q_1, b_1^+ \}\;
\{ r_1 \}\;
\{ s_1 \}\;
\{ t_1 \}\;
\{ u_1, a_1^- \}\;
\{ v_1, b_1^- \}\;
\{ d_1 \}
\\
&\qquad
\{ p_2, a_2^+ \}\;
\{ q_2, b_2^+ \}\;
\{ r_2 \}\;
\{ s_2 \}\;
\{ t_2 \}\;
\{ u_2, a_2^- \}\;
\{ v_2, b_2^- \}\;
\{ d_2 \}
\\\smallskip\\
\Gamma':
&\qquad
e_1^2
\quad
\hl{e_1^1 f_1^1}
\quad
f_1^2
\quad
z_1
\quad
e_2^2
\quad
\hl{e_2^1 f_2^1}
\quad
f_2^2
\quad
z_2
\quad
e_3^1
\quad
\hl{e_3^2 f_3^2}
\quad
f_3^1
\quad
z_3
\\
&\qquad
p_1
\quad
\true{a_1^+ b_1^+}
\quad
\true{q_1 r_1}
\quad
\true{s_1 t_1}
\quad
a_1^-
\quad
\true{u_1 v_1}
\quad
b_1^-
\quad
d_1
\\
&\qquad
a_2^+
\quad
\false{p_2 q_2}
\quad
b_2^+
\quad
\false{r_2 s_2}
\quad
\false{t_2 u_2}
\quad
\false{a_2^- b_2^-}
\quad
v_2
\quad
d_2
\\\smallskip\\
\Pi:
&\qquad
\{ p_1, e_1^1 \}\;
\{ q_1, f_1^1 \}\;
\{ r_1, e_2^1 \}\;
\{ s_1, f_2^1 \}\;
\{ t_1, d_1 \}\;
\{ u_1, e_3^1 \}\;
\{ v_1, f_3^1 \}\;
\{ a_1^+, a_1^- \}\;
\{ b_1^+, b_1^- \}
\\
&\qquad
\{ p_2, e_1^2 \}\;
\{ q_2, f_1^2 \}\;
\{ r_2, d_2 \}\;
\{ s_2, e_2^2 \}\;
\{ t_2, f_2^2 \}\;
\{ u_2, e_3^2 \}\;
\{ v_2, f_3^2 \}\;
\{ a_2^+, a_2^- \}\;
\{ b_2^+, b_2^- \}
\\
&\qquad
\{ z_1 \}
\quad
\{ z_2 \}
\quad
\{ z_3 \}
\\\smallskip\\
\Pi':
&\qquad
p_1
\quad
\hl{e_1^1 f_1^1}
\quad
\true{q_1 r_1}
\quad
\hl{e_2^1 f_2^1}
\quad
\true{s_1 t_1}
\quad
d_1
\quad
e_3^1
\quad
\true{u_1 v_1}
\quad
f_3^1
\quad
a_1^-
\quad
\true{a_1^+ b_1^+}
\quad
b_1^-
\\
&\qquad
e_1^2
\quad
\false{p_2 q_2}
\quad
f_1^2
\quad
d_2
\quad
\false{r_2 s_2}
\quad
e_2^2
\quad
f_2^2
\quad
\false{t_2 u_2}
\quad
\hl{e_3^2 f_3^2}
\quad
v_2
\quad
a_2^+
\quad
\false{a_2^- b_2^-}
\quad
b_2^+
\\
&\qquad
z_1 \quad z_2 \quad z_3
\end{align*}
\caption{The two weak orders $\Gamma$ and $\Pi$ and their linear extensions $\Gamma'$ and $\Pi'$
corresponding to the \sat{3}{2} instance
$c_1 = x_1 \lor x_2$, $c_2 = x_1 \lor \bar x_2$,
$c_3 = \bar x_1 \lor \bar x_2$
and the assignment $x_1 = \mathrm{true}$ and $x_2 = \mathrm{false}$.}
\label{fig:sat}
\end{figure}

This completes the construction.
Refer to Figure~\ref{fig:sat} for an example.
Note that every bucket contains at most two markers.

\begin{lemma}\label{lem:ww}
There exists an assignment of $X$
satisfying at least $k$ clauses in $C$
if and only if
$\Gamma$ and $\Pi$
admit linear extensions
with at least $4n + k$ adjacencies.
\end{lemma}

\begin{proof}
We first prove the direct implication.
Suppose there exists an assignment of $X$
satisfying at least $k$ clauses in $C$.
We will linearize $\Gamma$ and $\Pi$
by ordering the markers in their buckets.
For each $i$, $1 \le i \le n$,
consider two cases:
\begin{itemize}\setlength\itemsep{0pt}

\item
$x_i$ is true.
\begin{itemize}\setlength\itemsep{0pt}

\item
Order the markers in the six buckets
$Q_i,R_i,S_i,T_i,U_i,V_i$
and correspondingly in the six buckets
$Q'_i,R'_i,S'_i,T'_i,U'_i,V'_i$
such that
$q_i r_i, s_i t_i, u_i v_i$ are three adjacencies.

\item
Order the markers in $A_i$, $B_i$, and $P'_i$
such that
$a_i^+ b_i^+$
is an adjacency.

\item
If $x_i^1 = c_j^h$,
and if the markers in $E_j$ and $F_j$ are still unordered,
order the markers in $E_j$, $F_j$, and $P_i$
such that
$e_j^h f_j^h$
is an adjacency.

\item
If $x_i^2 = c_j^h$ is positive,
and if the markers in $E_j$ and $F_j$ are still unordered,
order the markers in $E_j$ and $F_j$
such that
$e_j^h f_j^h$
is an adjacency.

\end{itemize}

\item
$x_i$ is false.
\begin{itemize}\setlength\itemsep{0pt}

\item
Order the markers in the six buckets
$P_i,Q_i,R_i,S_i,T_i,U_i$
and correspondingly in the six buckets
$P'_i,Q'_i,R'_i,S'_i,T'_i,U'_i$
such that
$p_i q_i, r_i s_i, t_i u_i$ are three adjacencies.

\item
Order the markers in $A_i$, $B_i$, and $V'_i$
such that
$a_i^- b_i^-$
is an adjacency.

\item
If $x_i^3 = c_j^h$,
and if the markers in $E_j$ and $F_j$ are still unordered,
order the markers in $E_j$, $F_j$, and $V_i$
such that
$e_j^h f_j^h$
is an adjacency.

\item
If $x_i^2 = c_j^h$ is negative,
and if the markers in $E_j$ and $F_j$ are still unordered,
order the markers in $E_j$ and $F_j$
such that
$e_j^h f_j^h$
is an adjacency.

\end{itemize}

\end{itemize}
For any bucket with two markers that are still unordered,
order them arbitrarily.
It is easy to check that
between the two linear extensions thus obtained,
there are $4n+k$ adjacencies in total,
including $3n$ adjacencies between variable markers,
$n$ adjacencies between selection markers,
and $k$ adjacencies between literal markers.

\bigskip
We next prove the reverse implication.
The following properties can be easily verified
for any two linear extensions of $\Gamma$ and $\Pi$, respectively:
\begin{enumerate}\setlength\itemsep{0pt}

\item[1.]
The only possible adjacencies are from the following sets:
\begin{align*}
\langle ef \rangle_j &= 
\{\, e_j^1 f_j^1, \, e_j^2 f_j^2 \,\},
&&1 \le j \le m,
\\
\langle ab \rangle_i &= 
\{\, a_i^+ b_i^+, \, a_i^- b_i^- \,\},
&&1 \le i \le n,
\\
\langle pqrstuv \rangle_i &= 
\{\, p_i q_i, \, q_i r_i, \, r_i s_i, \, s_i t_i, \, t_i u_i, \, u_i v_i \,\},
&&1 \le i \le n.
\end{align*}

\item[2.]
For each $j$, $1 \le j \le m$,
there is at most one adjacency from $\langle ef \rangle_j$.

\item[3.]
For each $i$, $1 \le i \le n$,
there is at most one adjacency from $\langle ab \rangle_i$.

\end{enumerate}
Observe that each of the seven buckets
$P_i,Q_i,R_i,S_i,T_i,U_i,V_i$
contains either a literal marker or a dummy marker besides a variable marker.
This implies the following property:
\begin{enumerate}\setlength\itemsep{0pt}

\item[4.]
For each $i$, $1 \le i \le n$,
each of the seven variable markers $p_i,q_i,r_i,s_i,t_i,u_i,v_i$
can participate in at most one adjacency
in $\langle pqrstuv \rangle_i$,
and hence there are at most three adjacencies from
$\langle pqrstuv \rangle_i$.

\end{enumerate}
Properties 3 and 4 together imply the following property:
\begin{enumerate}\setlength\itemsep{0pt}

\item[5.]
For each $i$, $1 \le i \le n$,
there are
at most four adjacencies from
$\langle ab \rangle_i \cup \langle pqrstuv \rangle_i$.
Moreover, if there are exactly four adjacencies from
$\langle ab \rangle_i \cup \langle pqrstuv \rangle_i$,
then these adjacencies must be either
$\{\, a_i^+ b_i^+, q_i r_i, \, s_i t_i, \, u_i v_i \,\}$
or
$\{\, p_i q_i, \, r_i s_i, \, t_i u_i, \, a_i^- b_i^- \,\}$.
\begin{itemize}\setlength\itemsep{0pt}

\item
If
$\{\, a_i^+ b_i^+, q_i r_i, \, s_i t_i, \, u_i v_i \,\}$
are the four adjacencies,
then
\begin{itemize}\setlength\itemsep{0pt}
\item
there can be an adjacency between
the two literal markers $e_j^h$ and $f_j^h$
in $P_i$ and $Q_i$
corresponding to the positive literal $x_i^1$, and
\item
there can be an adjacency between
the two literal markers $e_j^h$ and $f_j^h$
in $R_i$ and $S_i$
corresponding to the literal $x_i^2$,
if it is also positive.
\end{itemize}

\item
If
$\{\, p_i q_i, \, r_i s_i, \, t_i u_i, \, a_i^- b_i^- \,\}$
are the four adjacencies,
then
\begin{itemize}\setlength\itemsep{0pt}
\item
there can be an adjacency between
the two literal markers $e_j^h$ and $f_j^h$
in $U_i$ and $V_i$
corresponding to the negative literal $x_i^3$, and
\item
there can be an adjacency between
the two literal markers $e_j^h$ and $f_j^h$
in $S_i$ and $T_i$
corresponding to the literal $x_i^2$,
if it is also negative.
\end{itemize}

\end{itemize}

\end{enumerate}

With respect to two linear extensions of $\Gamma$ and $\Pi$, respectively,
we say that a literal $x_i^g = c_j^h$,
$1 \le i \le n$,
$1 \le g \le 3$, $1 \le j \le m$, $1 \le h \le 2$,
is \emph{realized} if $e_j^h f_j^h$ is an adjacency of the two linear extensions,
and we say that the two linear extensions are \emph{consistent} for a variable $x_i$,
$1 \le i \le n$,
if the realized literals among $x_i^1,x_i^2,x_i^3$
are either all positive or all negative.

\begin{figure}[htbp]
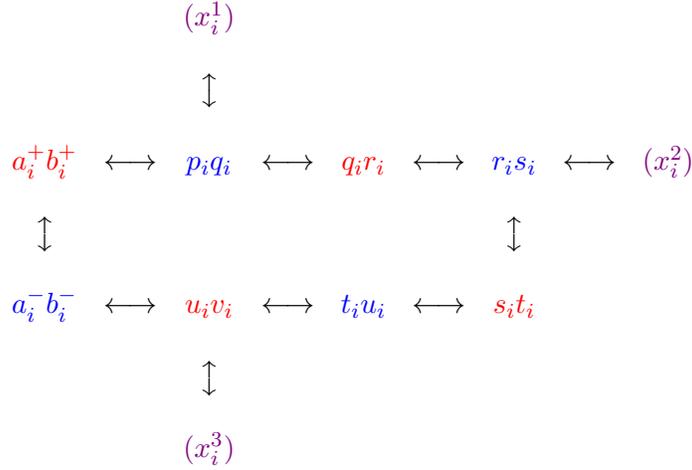

$$
\begin{array}{ccccccccc}
& & \hl{(x_i^1)} & & & & & &
\\
& & & & & & & &
\\
& & \big\updownarrow & & & & & &
\\
& & & & & & & &
\\
\true{a_i^+ b_i^+}
& \longleftrightarrow &
\false{p_i q_i}
& \longleftrightarrow &
\true{q_i r_i}
& \longleftrightarrow &
\false{r_i s_i}
& \longleftrightarrow &
\hl{(x_i^2)}
\\
& & & & & & & &
\\
\big\updownarrow & & & & & & \big\updownarrow & &
\\
& & & & & & & &
\\
\false{a_i^- b_i^-}
& \longleftrightarrow &
\true{u_i v_i}
& \longleftrightarrow &
\false{t_i u_i}
& \longleftrightarrow &
\true{s_i t_i}
& &
\\
& & & & & & & &
\\
& & \big\updownarrow & & & & & &
\\
& & & & & & & &
\\
& & \hl{(x_i^3)} & & & & & &
\end{array}
$$
\caption{Conflicts between candidate adjacencies.}
\label{fig:conflict}
\end{figure}

We say that two candidate adjacencies have a \emph{conflict} if they cannot be
both present.
Refer to Figure~\ref{fig:conflict} for the conflict graph
with edges between conflicting adjacencies
among
the eight adjacencies
$a_i^+ b_i^+$,
$a_i^- b_i^-$,
$p_i q_i$,
$q_i r_i$,
$r_i s_i$,
$s_i t_i$,
$t_i u_i$,
$u_i v_i$,
and the three adjacencies $e_j^h f_j^h$ corresponding
to the three literals $x_i^1,x_i^2,x_i^3$ of the variable $x_i$,
$1 \le i \le n$.

If $x_i^2$ is positive,
then there are three possible causes for inconsistency:
\begin{itemize}\setlength\itemsep{0pt}

\item
The three adjacencies for $x_i^1$, $x_i^2$, and $x_i^3$ are all realized.
Then $p_i q_i$, $r_i s_i$, and $u_i v_i$ cannot be realized.
Among $a_i^- b_i^-$, $a_i^+ b_i^+$,
$q_i r_i$, $s_i t_i$, and $t_i u_i$,
at most three can be realized.
We can reorder the markers of the buckets in
$\langle X_i \rangle$
and
$\langle Y_i \rangle$
such that the four adjacencies
$a_i^+ b_i^+$, $q_i r_i$, $s_i t_i$, $u_i v_i$,
and
the two adjacencies for
$x_i^1$ and $x_i^2$ are realized.

\item
The two adjacencies for $x_i^1$ and $x_i^3$ are realized,
and the adjacency for $x_i^2$ is not realized.
Then $p_i q_i$ and $u_i v_i$ cannot be realized.
Among $a_i^- b_i^-$, $a_i^+ b_i^+$,
$q_i r_i$, $r_i s_i$, $s_i t_i$, and $t_i u_i$,
at most three can be realized.
We can reorder the markers of the buckets in
$\langle X_i \rangle$
and
$\langle Y_i \rangle$
such that the four adjacencies
$a_i^+ b_i^+$, $q_i r_i$, $s_i t_i$, $u_i v_i$,
and
the adjacency for
$x_i^1$ are realized.

\item
The two adjacencies for $x_i^2$ and $x_i^3$ are realized,
and the adjacency for $x_i^1$ is not realized.
Then $r_i s_i$ and $u_i v_i$ cannot be realized.
Among $a_i^- b_i^-$, $a_i^+ b_i^+$,
$p_i q_i$, $q_i r_i$, $s_i t_i$, and $t_i u_i$,
at most three can be realized.
We can reorder the markers of the buckets in
$\langle X_i \rangle$
and
$\langle Y_i \rangle$
such that the four adjacencies
$a_i^+ b_i^+$, $q_i r_i$, $s_i t_i$, $u_i v_i$,
and
the adjacency for
$x_i^2$ are realized.

\end{itemize}

In each case, we can avoid inconsistency and realize
the same number of adjacencies.
The situation is similar if $x_i^2$ is negative.

Now suppose that there is no inconsistency.
Among the eight adjacencies
$a_i^- b_i^-$, $a_i^+ b_i^+$,
$p_i q_i$, $q_i r_i$, $r_i s_i$, $s_i t_i$, $t_i u_i$, $u_i v_i$,
there can be at most four adjacencies.
We can reorder the markers of the buckets in
$\langle X_i \rangle$
and
$\langle Y_i \rangle$
to realize exactly four adjacencies:
if no negative literal of $x_i$ is realized,
then realize
$a_i^+ b_i^+$,
$q_i r_i$, $s_i t_i$, $u_i v_i$;
otherwise,
realize
$a_i^- b_i^-$,
$p_i q_i$, $r_i s_i$, $t_i u_i$.

Suppose there exist two linear extensions of $\Gamma$ and $\Pi$, respectively,
with at least $4n + k$ adjacencies.
Then by the above analysis,
we can assume that the two linear extensions are consistent
and moreover realize exactly four adjacencies between
$\langle X_i \rangle$
and
$\langle Y_i \rangle$ for each $i$, $1 \le i \le n$.
Then the remaining at least $k$ adjacencies must be adjacencies of literal markers,
between $\langle C_j \rangle$ and $\langle Y_i \rangle$.
Assign each variable $x_i$ to true if $a_i^+ b_i^+$ is realized,
and to false if $a_i^- b_i^-$ is realized.
Then the at least $k$ adjacencies between literal markers
must correspond to at least $k$ satisfied clauses.
\end{proof}

We are now ready to present two L-reductions from \sat{3}{2} to
\maxadj\ww\ and \minbrk\ww, respectively.
Recall~\eqref{eq:f} and~\eqref{eq:g}.

Let $k^*$ be the maximum number of satisfied clauses in the given \sat{3}{2} instance.
By a naive assignment in which all $n$ variables are true,
we can make sure that all positive literals are true.
Recall that for each variable there is at least one positive literal,
and each clause contains exactly two literals.
Thus there are at least $n$ positive literals,
and at least $n/2$ clauses containing at least one positive literal,
which are satisfied.
Thus $k^* \ge \frac12 n$, and hence $n \le 2 k^*$.

In both L-reductions,
the polynomial-time function $f$
is just the construction described above which satisfies Lemma~\ref{lem:ww}.
By the direct implication of Lemma~\ref{lem:ww},
the maximum number of adjacencies is $4n + k^* \le 8k^* + k^* = 9 k^*$.
Thus we can set $\alpha = 9$ for the L-reduction to \maxadj\ww.

By double counting the number of variable-clause incidences,
we have $3n = 2m$.
Since the number of markers is $12n + 5m$,
the minimum number of breakpoints is
$(12n + 5m - 1) - (4n + k^*)
\le 8n + 5m - k^*
= 8n + 5\cdot\frac32 n - k^*
= \frac{31}2 n - k^*
\le 31 k^* - k^*
= 30 k^*$.
Thus we can set $\alpha = 30$ for the L-reduction to \minbrk\ww.

For the other direction,
consider any solution to the reduced instance,
with $4n + k$ adjacencies and $(12n + 5m - 1) - (4n + k)$ breakpoints.
By the reverse implication of Lemma~\ref{lem:ww},
there is a function $g$ that converts this solution
to a variable assignment for the \sat{3}{2} instance
that satisfies at least $k$ clauses.
The absolute error of the solution, for both \maxadj\ww\ and \minbrk\ww,
is exactly $|k^* - k|$.
The absolute error of the converted solution for \sat{3}{2}
is at most $|k^* - k|$.
Thus we can set $\beta = 1$ for both L-reductions.

Since \sat{3}{2} is APX-hard, it follows by the two L-reductions that
the two optimization problems \maxadj\ww\ and \minbrk\ww\ are APX-hard too.
This completes the proof of Theorem~\ref{thm:ww}.

\section{Exact algorithm for aligning a linear order and a weak order}

In this section we prove Proposition~\ref{prp:lw}.
Let $\Sigma$ be a set of $n$ markers.
Let $\Gamma$ be a linear order on $\Sigma$.
Let $\Pi$ be a weak order on $\Sigma$,
represented by a partition of $\Sigma$ into $k$ buckets
$\Sigma_1,\Sigma_2,\ldots,\Sigma_k$.

To linearize $\Pi$,
first partition the markers in each bucket $\Sigma_h$, $1 \le h \le k$,
into maximal blocks of markers that appear as a contiguous substring of $\Gamma$,
and denote these blocks by $B_h$.
By a simple replacement argument, we can assume without loss of generality
that, in computing a linear extension of $\Pi$ to maximize the number of adjacencies
with $\Gamma$,
the markers in each block always appear consecutively as a single unit,
in the same order as the corresponding substring of $\Gamma$.

For each $i$, $1 \le i \le k$,
and for each block $b \in B_i$,
denote by $m_\textsl{adj}(i, b)$ the maximum number of adjacencies
between $\Gamma$ and any linear extension of the weak order
represented by $\Sigma_1,\ldots,\Sigma_i$,
obtained by independently ordering the blocks in $B_h$ for each bucket $\Sigma_h$,
$1 \le h \le i$,
with the additional constraint that $b$ is the last block in $B_i$.
The table $m_\textsl{adj}(i, b)$ can be computed
by dynamic programming in polynomial time.
Then
$$
n_\textsl{adj} = \max\{\, m_\textsl{adj}(k, b) \mid b \in B_k \,\}
$$
gives the maximum number of adjacencies,
and correspondingly $n_\textsl{brk} = n - 1 - n_\textsl{adj}$
gives the minimum number of breakpoints,
between the linear order $\Gamma$ and the weak order $\Pi$.

Thus we have a polynomial-time exact algorithm for
\maxadj\lw\ and \minbrk\lw.
This completes the proof of Proposition~\ref{prp:lw}.

\section{An open question}

Is there a polynomial-time exact algorithm for \maxadj\ls\ and \minbrk\ls?

\end{document}